\documentclass{amsart}
\usepackage{amssymb, color}

\newtheorem*{theorem}{Theorem}
\newtheorem*{corollary}{Corollary}

\begin{document}
\title[On symmetries of elasticity tensors and Christoffel matrices]
{On symmetries of elasticity tensors and Christoffel matrices}
\author[A. B\'{o}na]{Andrej B\'{o}na}
\address{Department of Exploration Geophysics, Curtin Technical University, Perth, Australia}
\email{a.bona@curtin.edu.au}
\author[\c{C}. D\.{i}ner]{\c{C}a\u{g}r\i \ D\.{i}ner}
\address{Geophysics Department, Kandilli Observatory and Earthquake Research Institute, Bo\u{g}azi\c{c}i University, Istanbul, Turkey}
\email{cadiner@yahoo.com}
\author[M. Kochetov]{Mikhail Kochetov}
\address{Department of Mathematics and Statistics, Memorial University of Newfoundland, St.~John's, Newfoundland, Canada}
\email{mikhail@mun.ca}
\author[M. A. Slawinski]{Michael A. Slawinski}
\address{Department of Earth Sciences, Memorial University of Newfoundand, St.~John's, Newfoundland, Canada}
\email{mslawins@mun.ca}



\begin{abstract}
We prove that the symmetry group of an elasticity tensor is equal to the symmetry group of the corresponding Christoffel matrix.
\end{abstract}

\maketitle

\section{Introduction}
The Curie  principle states that the results are at least as symmetric as the causes. It is illustrated for wave phenomena by the fact that wavefronts propagating in a Hookean solid from a point source are at least as symmetric as the solid itself. In other words, the material symmetries of the elasticity tensor of that solid are a subgroup --- possibly proper --- of wavefront symmetries, as shown by B\'{o}na et al.~\cite{BBS}. Herein, we prove that the elasticity tensor and the Christoffel matrix, from which wavefronts are derived, have the same symmetry groups, as suggested by B\'{o}na et al.\  \cite{BBS}, but for reasons more subtle than presented therein. Thus, it follows that the increase of symmetry may occur in passing from the Christoffel matrix to wavefronts, not in passing from the elasticity tensor to the Christoffel matrix. The importance of this result is our being able to obtain information about material symmetries of a Hookean solid from measurements that are common in seismology, which allow us to estimate the Christoffel matrix, not the elasticity tensor directly.

We begin this paper with a brief description of relations between the elasticity tensor and the Christoffel matrix and between the matrix and the wavefronts. A more detailed description can be found in many sources, notably, {\v C}erven{\'y} \cite{cerveny}; the description below follows Bos and Slawinski \cite{BS}.

\section{Background}

A Hookean solid is defined by the following constitutive equation:
\begin{equation}\label{eq:Hooke}
\sigma_{ij}=\sum\limits _{k,\ell=1}^{3}c_{ijk\ell}\varepsilon_{k\ell},
\end{equation}
where $c$ is the elasticity tensor relating the stress tensor, $\sigma$, and the strain tensor, $\varepsilon$. As shown by Forte and Vianello \cite{FV}, $c$ belongs to one of eight symmetry groups; the symmetry group of $c$ is the set of orthogonal  transformations under which $c$ is invariant. Also, $c$ possesses index symmetries: $c_{ijk\ell}=c_{jik\ell}=c_{k\ell ij}$.

Inserting constitutive equation (\ref{eq:Hooke}) into the equations of motion,
\begin{equation}\label{EqMot}
\sum\limits _{j=1}^{3}\frac{\partial\sigma_{ij}}{\partial x_j}=\rho\frac{\partial^2u_i}{\partial t^2},
\end{equation}
where $i\in\left\{ 1,2,3\right\} $, $\rho$ is the mass density, $x$ and $t$ are the spatial and temporal variables, respectively, and $u$ is the displacement vector whose components are related to the strain tensor by $\varepsilon_{ij}:=(\partial u_i/\partial x_j+\partial u_j/\partial x_i)/2$,
we obtain the equations of motion in a Hookean solid,
\begin{equation}\label{ElDynEq}
\rho\left(x\right)\frac{\partial^{2}u_{i}\left(x,t\right)}{\partial t^{2}}=\sum\limits _{j,k,\ell=1}^{3}\left(\frac{\partial c_{ijk\ell}\left(x\right)}{\partial x_{j}}\frac{\partial u_{k}\left(x,t\right)}{\partial x_\ell}+c_{ijk\ell}\left(x\right)\frac{\partial^{2}u_{k}\left(x,t\right)}{\partial x_{j}\partial x_\ell}\right),
\end{equation}
where $i\in\left\{ 1,2,3\right\} $, which
describe propagation of displacement in an anisotropic inhomogeneous
elastic medium.
Wavefronts and polarizations of waves propagating in such a medium stem from the characteristic equation of these equations, which is
\begin{equation}\label{DetVan}
\det\left(\sum\limits _{j,\ell=1}^{3}c_{ijk\ell}n_{j}n_\ell-\frac{\rho}{v^2}\delta_{ik}\right)=0\text{,\qquad}i,k\in\;\left\{ 1,2,3\right\},
\end{equation}
where $n$ is the unit vector normal to the wavefront and $v$ is the wavefront velocity. It is common to let $\Gamma_{ik}(n):=c_{ijk\ell}n_jn_\ell$, which is the Christoffel matrix. The symmetry group of the Christoffel matrix is the set of orthogonal transformations, $A$, such that $\Gamma(Au)(Av,Aw)=\Gamma(u)(v,w)$, for all $u,v,w$ in $\mathbb{R}^3$.
Note that $\Gamma$ is quadratic in the first variable, $u$; also note that $u$ and $v$ herein and in the next section are generic variables, not the physical entities of equations (\ref{EqMot}), (\ref{ElDynEq}) and (\ref{DetVan}). Also, the Christoffel matrix possesses index symmetries, $\Gamma_{ij}=\Gamma_{ji}$, which are inherited from the index symmetries of $c$. Equation (\ref{DetVan}) is a third-degree polynomial in $1/v^2$. Its root are the eigenvalues of $\Gamma$ scaled by $\rho$. Each of the three eigenvalues corresponds to the velocity of one of the three types of waves that propagate in a Hookean solid, and each eigenvector to its polarization.

As suggested by B\'{o}na et al.\  \cite{BBS}, the increase of symmetries between the Hookean solid and the wavefronts propagating within it can occur only between the Christoffel matrix and the wavefronts, which is tantamount to the equality of the symmetry groups of the tensor and the matrix.
This equality is proven in the next section. The proof is based on the fact that the Christoffel matrix determines uniquely the elasticity tensor, as shown by  B\'{o}na et al.\  \cite{BBS}, and also on the fact that this one-to-one correspondence respects the action of $O(3)$, which guarantees the equality of the symmetry groups. 

\section{Equality of symmetry groups}
Explicitly, the correspondence between $c$ and $\Gamma$ is
\begin{equation}
 c_{iki\ell}=\Gamma(e_i)(e_k,e_\ell) \label{ikil}
 \end{equation}
and
\begin{equation}
c_{iij\ell}=\frac{1}{2}\left[\Gamma(e_i+e_j)(e_i,e_\ell)-\Gamma(e_i-e_j)(e_i,e_\ell)-2\Gamma(e_i)(e_j,e_\ell)\right],\label{iijl}
\end{equation}
which are equations (3.7) and (3.8) in B\'{o}na et al.\  \cite{BBS}.

For a given $\Gamma$, these equations allow us to determine the components of $c$ due to repetitions among $i,j,k,\ell$, but they have different forms for different components. Thus, it is not obvious that the relation between $c$ and $\Gamma$ respects the action of $O(3)$. To establish this property, we restate this relation in a coordinate-free form. We can regard $c$ as a quadrilinear form,
$c:\mathbb{R}^3\times\mathbb{R}^3\times\mathbb{R}^3\times\mathbb{R}^3\rightarrow\mathbb{R}^3$,
and $\Gamma$ as a function,
$\Gamma:\mathbb{R}^3\times\mathbb{R}^3\times\mathbb{R}^3\rightarrow\mathbb{R}^3$,
which is quadratic in the first variable and linear in the other two. Let us linearize $\Gamma$ by defining
\begin{equation}
\tilde\Gamma(x,y,v,w):=\frac{1}{2}\left[\Gamma(x+y)(v,w)-\Gamma(x)(v,w)-\Gamma(y)(v,w)\right],
\label{eq:GammaTilde}
\end{equation}
for all $x,y,v,w\in\mathbb{R}^3$. Then $\Gamma$ can be recovered from $\tilde\Gamma$ as follows: 
\begin{equation}
\Gamma(u)(v,w)=\tilde\Gamma(u,u,v,w).
\label{eq:GammaGamma}
\end{equation}

\begin{theorem}
For all $x,y,v,w\in\mathbb{R}^3$, we have
\begin{align}
c(x,y,v,w)&=\tilde\Gamma(x,v,y,w)+\tilde\Gamma(y,v,x,w)-\tilde\Gamma(x,y,v,w);\label{eq:c_GammaTilde} \\
\tilde\Gamma(x,y,v,w)&=\frac12[c(x,v,y,w)+c(y,v,x,w)]\label{eq:GammaTilde_c}.
\end{align}
\end{theorem}

\begin{proof}
Equation (\ref{eq:GammaTilde}) can be restated as $\Gamma(x+y)(v,w)=\Gamma(x)(v,w)+\Gamma(y)(v,w)+2\tilde\Gamma(x,y,v,w)$. Hence we can rewrite formula (\ref{iijl}) as
\begin{equation}\label{iijlTilde}
c(e_i,e_i,e_j,e_\ell)=2\tilde\Gamma(e_i,e_j,e_i,e_\ell)-\Gamma(e_i)(e_j,e_\ell).
\end{equation}
In view of the linearity of $c$ and $\tilde\Gamma$ in each variable and the linearity of $\Gamma$ in the second and third variables, we obtain:
\begin{equation}\label{eq:half_way1}
c(e_i,e_i,v,w)=2\tilde\Gamma(e_i,v,e_i,w)-\Gamma(e_i)(v,w),
\end{equation}
for all $v,w\in\mathbb{R}^3$. Now suppose $k\ne i$. Then we get:
\begin{equation}\label{eq:half_way2}
c(e_k,e_k,v,w)=2\tilde\Gamma(e_k,v,e_k,w)-\Gamma(e_k)(v,w),
\end{equation}
and, replacing $\{e_i,e_k\}$ by $\{(e_i+e_k)/\sqrt{2},-(e_i-e_k)/\sqrt{2}\}$ in the orthonormal basis, we also get
\begin{equation}\label{eq:half_way3}
c(e_i+e_k,e_i+e_k,v,w)=2\tilde\Gamma(e_i+e_k,v,e_i+e_k,w)-\Gamma(e_i+e_k)(v,w).
\end{equation}
Subtracting equations (\ref{eq:half_way1}) and (\ref{eq:half_way2}) from equation (\ref{eq:half_way3}) and using the index symmetry of $c$, we obtain
\begin{equation}\label{eq:Misha14}
c(e_i,e_k,v,w)=\tilde\Gamma(e_i,v,e_k,w)+\tilde\Gamma(e_k,v,e_i,w)-\tilde\Gamma(e_i,e_k,v,w).
\end{equation}
Equations (\ref{eq:half_way1}) and (\ref{eq:Misha14}) imply equation (\ref{eq:c_GammaTilde}) in view of the linearity of both sides of equation (\ref{eq:c_GammaTilde}).

Note that by substituting $x=v=e_i$, $y=e_k$ and $w=e_\ell$ in equation (\ref{eq:c_GammaTilde}), we obtain expression (\ref{ikil}), and by substituting $x=y=e_i$, $v=e_j$ and $w=e_\ell$, we obtain expression (\ref{iijlTilde}), which is an equivalent form of expression (\ref{iijl}).

Finally, equation (\ref{eq:GammaTilde_c}) follows immediately from definition (\ref{eq:GammaTilde}) and equation (\ref{ikil}).
\end{proof}

The mappings $\tilde\Gamma\mapsto c$ and $c\mapsto\tilde\Gamma$ given by expressions (\ref{eq:c_GammaTilde}) and (\ref{eq:GammaTilde_c}), respectively, are linear isomorphisms that are inverses of one another. Clearly, these mappings respect the action of $O(3)$.

\begin{corollary}
The symmetries of the elasticity tensor, $c$, are the same as the symmetries of the Christoffel matrix, $\Gamma$.
\end{corollary}

\begin{proof}
Any symmetry of $\Gamma$ is a symmetry of $\tilde\Gamma$ and vice versa by equations (\ref{eq:GammaTilde}) and (\ref{eq:GammaGamma}), respectively. It remains to invoke the isomorphisms (\ref{eq:c_GammaTilde}) and (\ref{eq:GammaTilde_c}) between $\tilde\Gamma$ and $c$.
\end{proof}

\section{Conclusions}
As stated by the corollary, the material symmetries of the elasticity tensor, $c$, are indeed the same as the symmetries of the Christoffel matrix, $\Gamma$, as sugessted by B\'{o}na et al.\  \cite{BBS}. This means that the increase of symmetry occurs between the Christoffel matrix and the wavefronts, not between the elasticity tensor and the Christoffel matrix. As shown by B\'{o}na et al.\  \cite{BBS}, this result entails that the symmetry group of $c$ is the intersection of the symmetry groups of wavefronts and polarizations.
Otherwise --- if the increase of symmetry could occur between $c$ and $\Gamma$ --- the knowledge of the wavefront and polarization symmetries would be insufficient to infer the symmetry of the Hookean solid. In mathematical language, the equality of symmetry groups of $c$ and $\Gamma$ ensures that the knowledge of eigenvalues and eigenvectors, which allows us to reconstruct $\Gamma$, entails the knowledge of the symmetries of $c$.
Furthermore, measurements might result in computing $\Gamma$ without considering wavefronts and polarizations, as shown by Dewangan and Grechka \cite{Grechka}. According to the corollary, the material symmetries of the Hookean solid can be obtained directly from $\Gamma$.

Finally, the formul{\ae} relating $c$ and $\Gamma$ that are stated in the theorem are of interest, since --- unlike equations (3.7) and (3.8) in B\'{o}na et al.\  \cite{BBS} --- they have the same form for all components of $c$ and $\Gamma$.

\section*{Acknowledgements} 
M. Kochetov's  and M.A. Slawinski's research was supported by the Natural Sciences and Engineering Research Council of Canada.


\end{document}